\documentclass[letterpaper, 10pt, conference]{ieeeconf}  

\IEEEoverridecommandlockouts                              

\title{\LARGE\bf Decentralized Coherent Quantum Control Design for Translation Invariant Linear Quantum Stochastic Networks with Direct Coupling} 
\author{Arash Kh. Sichani, \qquad Igor G. Vladimirov, \qquad Ian R. Petersen
\thanks{This work is supported by the Australian Research Council. The authors are with UNSW Canberra, ACT 2600, Australia. {\tt arash\_kho@hotmail.com, igor.g.vladimirov@gmail.com, i.r.petersen@gmail.com}.}
}

\usepackage{mathptmx} 

\usepackage{graphicx} 
\usepackage{times} 
\usepackage{amsmath} 
\usepackage{amssymb}  
\usepackage{datetime}
\usepackage{comment}
\usepackage{tikz}
\usetikzlibrary{calc}
\usetikzlibrary{matrix}

\newtheorem{lem}{Lemma}
\newtheorem{thm}{Theorem}

\def\<{\leqslant}           
\def\>{\geqslant}           

\def\d{\partial}
\def\wh{\widehat}
\def\wt{\widetilde}
\def\mod{\mathrm{mod}}   
\def\Res{\mathop{\mathrm{Res}}} 
\def\Re{\mathrm{Re}}   

\def\cH{\mathcal{H}}   
\def\mA{\mathbb{A}}    
\def\mR{\mathbb{R}}    
\def\mC{\mathbb{C}}    

\def\Tr{\mathrm{Tr}}       
\def\rT{\mathrm{T}}        

\def\bE{\mathbf{E}}    

\def\[[[{[\![\![}
\def\]]]{]\!]\!]}

\def\bra{\langle }
\def\ket{\rangle }

\def\re{\mathrm{e}}        
\def\rd{\mathrm{d}}        

\def\sJ{\mathsf{J}}

\def\br{\mathbf{r}}
\def\x{\times}
\def\ox{\otimes}

\def\mG{{\mathbb G}}

\def\cW{\mathcal{W}}

\def\cX{\mathcal{X}}

\def\cQ{{\mathcal Q}}

\def\cA{\mathcal{A}}
\def\cB{\mathcal{B}}
\def\cE{\mathcal{E}}

\def\cS{\mathcal{S}}

\def\mU{\mathbb{U}}

\def\mH{{\mathbb H}}
\def\mS{{\mathbb S}}

\DeclareMathAlphabet{\bit}{OML}{cmm}{b}{it}

\onecolumn
\begin{document}
\maketitle
\thispagestyle{empty}

\begin{abstract}
This paper is concerned with coherent quantum control design for translation invariant networks of identical quantum stochastic systems subjected to external quantum noise. The network is modelled as an open quantum harmonic oscillator and is governed by a set of linear quantum stochastic differential equations. The dynamic variables of this quantum plant satisfy the canonical commutation relations. Similar large-scale systems can be found, for example, in quantum metamaterials and optical lattices. The problem under consideration is to design a stabilizing decentralized coherent quantum controller in the form of  another  translation invariant quantum system, directly coupled to the plant, so as to minimize a weighted mean square functional of the dynamic variables of the interconnected networks. We consider this problem in the thermodynamic limit of infinite network size and present first-order necessary conditions for optimality of the controller.
\end{abstract}
\section{INTRODUCTION}
Currently emerging technologies open up opportunities to synthesize artificial optical media known as quantum metamaterials; see, for example, \cite{Rakhmanov2008, Notomi2008, Zheludev11, Quach11, Zagoskin12}. These large-scale quantum networks with are engineered from complex unit cells and are effectively homogeneous (in the sense of translational symmetries) on the scale of relevant wavelengths, for example, in the microwave range. In natural solids, the quantum energy level configurations of the constituent atoms or molecules specifies the optical behaviour of the material. 

In contrast, the controllable resonant characteristics of the building elements of quantum metamaterials, such as the Josephson devices or optical cavities \cite{Rakhmanov2008, Notomi2008,Quach11}, determine the electromagnetic response of the quantum metamaterials. The quantum metamaterials are considered to be a promising approach to the implementation of quantum computer elements which can maintain quantum coherence over many cycles of their internal evolution \cite{Quach11, Zagoskin11}.

A framework for modelling and analysis of  a wide range of open quantum systems, including those arising in quantum metamaterials, is provided by
quantum stochastic differential equations (QSDEs) \cite{HP_1984,P_1992}. In QSDEs, the environment is modelled as a heat bath of external fields acting on a boson Fock space \cite{P_1992}. In particular, linear QSDEs represent the Heisenberg evolution of pairs of conjugate operators in a multi-mode open quantum harmonic oscillator which is coupled to the external bosonic fields. This framework also allows for robust  stability analysis for certain classes of perturbed open quantum systems which have been addressed, for example, in \cite{Ian2012robust,SVP_2015}.

The analysis of large-scale  quantum networks can be effectively reduced in the case when they are organized as a translation invariant interconnection of identical elements  with periodic boundary conditions (PBCs); see, for example, \cite{Quach11, VP2014, SVP_TI_2015}. The PBCs rely on negligibility of boundary effects in such a network consisting of a sufficiently large number of subsystems. This technique is used for lattice models of interacting particle systems in statistical physics (for example, in the Ising model of ferromagnetism \cite{Newman99}).

Coherent quantum feedback control is aimed at achieving robust stability and robust performance through measurement-free interconnection of quantum systems \cite{J2010diss,SVP_P_2015}. In this approach, the controller is another quantum system which is coupled to the quantum plant, for example, through a bilateral energy interconnection, known as direct coupling \cite{Zhang11}. In comparison with the more traditional measurement-based feedback control techniques, coherent control benefits from the preservation of quantum coherence within the network.

In this paper, both the quantum plant and the coherent quantum controller are modelled as large fragments of translation invariant networks of linear quantum stochastic systems endowed with the PBCs. The nodes of the networks are directly coupled to each other within a finite interaction range. This interconnection is governed by linear QSDEs based on the Hamiltonian and coupling parametrization of the corresponding multi-mode open quantum harmonic oscillator whose dynamic variables satisfy the canonical commutation relations (CCRs).
Following a similar approach used in the classical control theory, we employ spatial Fourier transforms in order to obtain a more tractable representation of the dynamics of the quantum feedback network; see, for example, \cite{Wall1978,Bamieh_2012, VP2014, SVP_TI_2015} and the references therein.

We consider a weighted mean square performance index for the stable plant-controller network in the thermodynamic limit (when the network size goes to infinity). In this framework, the decentralized control design problem is formulated as the minimization of the cost functional over the parameters of the controller and its coupling with the plant. By calculating the Fr{\' e}chet derivatives of the cost functional, we obtain first-order necessary conditions for optimality of the controller. In comparison with previous results on coherent quantum LQG control \cite{VP_2013a} using Gramians and related algebraic Lyapunov equations (ALEs), the optimality  conditions developed  in the present paper employ spatial spectral densities, which reflects the translation invariant nature of the underlying problem.

\section{NOTATION}\label{sec:not}

Unless specified otherwise,  vectors are organized as columns, and the transpose $(\cdot)^{\rT}$ acts on matrices with operator-valued entries as if the latter were scalars. For a vector $X$ of operators $X_1, \ldots, X_r$ and a vector $Y$ of operators $Y_1, \ldots, Y_s$, the commutator matrix is defined as an $(r\x s)$-matrix
$
    [X,Y^{\rT}]
    :=
    XY^{\rT} - (YX^{\rT})^{\rT}
$
whose $(j,k)$th entry is the commutator
$
    [X_j,Y_k]
    :=
    X_jY_k - Y_kX_j
$ of the operators $X_j$ and $Y_k$.
Also, $(\cdot)^{\dagger}:= ((\cdot)^{\#})^{\rT}$ denotes the transpose of the entry-wise operator adjoint $(\cdot)^{\#}$. In application to complex matrices,  $(\cdot)^{\dagger}$ reduces to the complex conjugate transpose  $(\cdot)^*:= (\overline{(\cdot)})^{\rT}$. Furthermore, $\mS_r$ and $\mA_r$
denote
the subspaces of real symmetric and real antisymmetric matrices of order $r$, respectively and $\mA_2$ is spanned by the matrix
$\sJ:={\scriptsize
		\begin{bmatrix}
			0 & 1\\ -1 & 0
		\end{bmatrix}
}$. Also, $i:= \sqrt{-1}$ denotes the imaginary unit, $I_r$ is the identity matrix of order $r$,
positive (semi-) definiteness of matrices is denoted by ($\succcurlyeq$) $\succ$, and $\ox$ is the tensor product of spaces or operators (for example, the Kronecker product of matrices). The spectral radius of a matrix $M$ is denoted by  $\br(M)$. The adjoints and  self-adjointness of linear operators acting on matrices is understood in the sense of the Frobenius inner product
$
    \bra M,N\ket
    :=
    \Tr(M^*N)
$ of real or complex matrices. 
The Kronecker delta is denoted by $\delta_{jk}$, and $\mU:= \{z\in \mC:\ |z|=1\}$ is the unit circle in the complex plain. The complex residue of a function $f$ about a point $z_0$ is denoted by $\Res_{z=z_0}f(z)$. Also, $\bE \xi := \Tr(\rho \xi)$ denotes the quantum expectation of a quantum variable $\xi$ (or a matrix of such variables) over a density operator $\rho$ which specifies the underlying quantum state. For matrices of quantum variables, the expectation is evaluated entry-wise.

\section{LINEAR QUANTUM STOCHASTIC SYSTEMS}\label{sec:system}
We consider a quantum stochastic system interacting with external boson fields \cite{HP_1984,P_1992}.  The system has $N$ subsystems with associated $2n$-dimensional vectors $X_0, \ldots, X_{N-1}$ of dynamic variables which satisfy the CCRs
\vspace{-2mm}\begin{equation}
\label{xCCR}
    [X, X^{\rT}] = 2i \bit{\Theta},
    \qquad
    \bit{\Theta}:= I_N \ox \Theta,
    \qquad
    X:=
    {\small\begin{bmatrix}
        X_0\\
        \vdots\\
        X_{N-1}
    \end{bmatrix}}.
\end{equation}
Here, $\bit{\Theta}$ is a block diagonal joint CCR matrix, where $\Theta \in \mA_{2n}$ is a nonsingular matrix.
The system variables evolve in time according to the QSDE
\vspace{-2mm}\begin{equation}
\label{dx}
  	\!\rd X
    \!\!=\!\!
    \Big(\!
      i[H ,X]
      \!-
      \!\frac{1}{2}
       \bit{B}\bit{J} \bit{B}^{\rT} \bit{\Theta}^{-1} X
    \Big)\!\rd t
    \!+\!\!
    \bit{B} \rd W,
    \qquad
    W
    :=
    {\scriptsize
    \begin{bmatrix}
           W_0\\
       \vdots \\
    	W_{N-1}
    \end{bmatrix}
    }.
\end{equation}
Here, $W_0, \ldots, W_{N-1}$ are $2m$-dimensional vectors of quantum Wiener processes with a  positive semi-definite It\^{o} matrix $\Omega \in \mH_{2m}$:
\begin{equation}
\label{WW}
    \rd W_j \rd W_k^{\rT}
    =
    \delta_{jk}\Omega \rd t,
    \qquad
    \Omega := I_{2m} + iJ,
    \qquad
    J := I_m \ox \sJ.
\end{equation}
Accordingly, $\bit{J}:= I_N \ox J$ is a block diagonal matrix,  and the matrix $\bit{B} \in \mR^{2nN\x 2mN}$ in (\ref{dx}) is related to a matrix $\bit{M} \in \mR^{2mN \x 2nN}$ of linear dependence of the system-field coupling operators on the system variables by
\begin{equation}
	\label{BM}
    \bit{B} := 2\bit{\Theta} \bit{M}^{\rT}.
\end{equation}
The term $-\frac{1}{2} \bit{B} \bit{J} \bit{B}^{\rT} \bit{\Theta}^{-1} X$ in the drift of the QSDE (\ref{dx}) is associated with the system-field interaction and results from evaluating the Gorini-Kossakowski-Sudarshan-Lindblad decoherence superoperator \cite{GKS_1976,L_1976} at the system variables. Also, $H$ is the Hamiltonian which describes the self-energy of the system and is usually represented as a function of the system variables.
In the case of an open quantum harmonic oscillator \cite{EB_2005,GZ_2004}, the Hamiltonian $H$ is a quadratic function of the system variables
\begin{equation}
\label{H0}
    H
    :=
    \frac{1}{2}
    X^{\rT} R X
    =
    \frac{1}{2}
    \sum_{j,k=0}^{N-1}
    X_j^\rT R_{jk} X_k,
\end{equation}
where the energy matrix $R:= (R_{jk})_{0\< j,k< N} \in \mS_{2nN}$ is formed from blocks $R_{jk}= R_{kj}^{\rT} \in \mR^{2n\x 2n}$.
By substituting (\ref{H0}) into (\ref{dx}) and using the CCRs (\ref{xCCR}), it follows that the QSDE takes the form of a linear QSDE
\begin{equation}
\label{dx1}
  \rd X = A X\rd t+ \bit{B} \rd W,
\end{equation}
where the system matrix $A \in \mR^{2nN\x 2nN}$ is given by
\begin{equation}\label{A}
    A
    :=
    2\bit{\Theta} R
    - \frac{1}{2} \bit{B}\bit{J} \bit{B}^{\rT} \bit{\Theta}^{-1} .
\end{equation}
We consider a $2qN$-dimensional  vector $Y$ of $q$-dimensional output fields $Y_0, \ldots, Y_{N-1}$ which is evolved by the linear QSDE
\begin{equation}
	\label{QSDEout}
	\rd Y = \bit{C} X \rd t+ \bit{D} \rd W,
	\qquad
    Y:=
    {\small\begin{bmatrix}
        Y_0\\
        \vdots\\
        Y_{N-1}
    \end{bmatrix}},
\end{equation}
with $\bit{C} \in \mR^{2qN \times 2nN}$ and $\bit{D} \in \mR^{2qN \times 2mN}$,
and satisfies the non-demolition condition \cite{B_1989} with respect to the dynamic variables in the sense that
\begin{equation}
	\label{Nondem}
	[X(t),Y(s)^{\rT}]=0,
    \qquad
    t\> s.
\end{equation}
Due to this non-demolition nature, the output fields can be interpreted as ideal measurements over the open quantum system. The condition (\ref{Nondem}) implies an algebraic relation between $\bit{C}$ and $\bit{D}$ in (\ref{QSDEout}):
\begin{align}
	\label{rCDB}
	\bit{\Theta}\bit{C}^\rT + \bit{B} \bit{J} \bit{D}^\rT=0.
\end{align}
Furthermore, the matrix $\bit{D} \bit{J} \bit{D}^\rT$ is an appropriate  submatrix of $\bit{J}$ from (\ref{WW}) (so that $q\< m$).

\section{LINEAR QUANTUM STOCHASTIC NETWORK WITH PERIODIC BOUNDARY CONDITIONS}\label{sec:Ring_Topology}
Suppose the open quantum harmonic oscillator of the previous section represents a fragment of a translation invariant network which is organised as a one-dimensional chain of identical linear quantum stochastic  systems numbered by $k=0,\ldots, N-1$. Each node of the network interacts with the corresponding external boson field, and hence, the joint network-field coupling matrix $\bit{M}$ in (\ref{BM}) and the feedthrough matrix $\bit{D}$ in (\ref{QSDEout}) are block diagonal:
\begin{align}
	\label{MMDD}
    \bit{M} := I_N \ox M,
	\qquad \bit{D} = I_N \ox D,
\end{align}
where $M \in \mR^{2m\x 2n}$ and $D \in \mR^{2q \times 2m}$. Hence, the matrix $\bit{B}$ is block diagonal and so also is the matrix $\bit{C}$ in view of (\ref{rCDB}), (\ref{MMDD}) and nonsingularity of the CCR matrix $\Theta$ in (\ref{xCCR}):
\begin{equation}
\label{BC}
  \!\!\bit{B} = I_N \ox B,
  \quad
  \bit{C} = I_N \ox C,
  \quad
  B := 2\Theta M^{\rT},
  \quad
  C := 2DJ M.\!\!\!\!\!
\end{equation}
The nodes in the network are directly coupled to each other within a finite interaction range $d$. The fragment of the chain is assumed to be  large enough in the sense that $N > 2d$, and is  endowed with the PBCs, thus having a ring topology. A particular case of nearest neighbour interaction (when $d=1$) is depicted in Fig.~\ref{fig:Ring_Topology}.
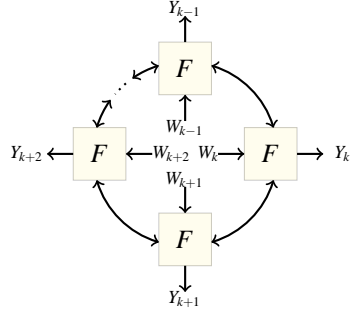
\begin{figure}[htp]
	\begin{center}
	\begin{tikzpicture}[scale=.35]
		\node at (0,1) {\scriptsize $W_{k-1}$};
		\node at (0.90,0) {\scriptsize $W_k$};
		\node at (0,-1) {\scriptsize $W_{k+1}$};
		\node at (-.5,0) {\scriptsize $W_{k+2}$};
		\draw [->, thick] (1.25,0) -- (2.25,0);
		\draw [->, thick] (-1.25,0) -- (-2.25,0);
		\draw [->, thick] (0,1.25) -- (0,2.25);
		\draw [->, thick] (0,-1.25) -- (0,-2.25);		
		\draw [->, thick] (4.25,0) -- (5.25,0);
		\draw [->, thick] (-4.25,0) -- (-5.25,0);
		\draw [->, thick] (0,4.25) -- (0,5.25);
		\draw [->, thick] (0,-4.25) -- (0,-5.25);
		\node at (0,5.5) {\scriptsize $Y_{k-1}$};
		\node at (-6,0) {\scriptsize $Y_{k+2}$};
		\node at (0,-5.5) {\scriptsize $Y_{k+1}$};
		\node at (6,0) {\scriptsize $Y_k$};		
		\draw [fill=yellow!40, opacity=.2] (2.25,1) rectangle (4.25,-1);
		\draw [fill=yellow!40, opacity=.2] (-1,4.25) rectangle (1,2.25);
		\draw [fill=yellow!40, opacity=.2] (-2.25,-1) rectangle (-4.25,1);
		\draw [fill=yellow!40, opacity=.2] (1,-2.25) rectangle (-1,-4.25);
		\node at (3.25,0) {$F$};
		\node at (0,3.25) {$F$};
		\node at (-3.25,0) {$F$};
		\node at (0,-3.25) {$F$};
		
		\draw [<->, thick] ([shift=(16.5:3.5)] 0,0) arc (16.5:73.5:3.5); 
		\draw [<->, thick] ([shift=(106.5:3.5)] 0,0) arc (106.5:125:3.5); 
		\draw [dotted, thick] ([shift=(130:3.5)] 0,0) arc (130.5:140:3.5);
		\draw [<->, thick] ([shift=(144:3.5)] 0,0) arc (144:163.5:3.5);
		\draw [<->, thick]([shift=(196.5:3.5)] 0,0) arc (196.5:253.5:3.5); 
		\draw [<->, thick]([shift=(-16.5:3.5)] 0,0) arc (-16.5:-73.5:3.5); 
	\end{tikzpicture}
	\end{center}\vskip-5mm
\caption{A finite fragment of the translation invariant network of open quantum systems with direct coupling between the nearest neighbours. Also shown are the input and output fields of the nodes of the network.}
\label{fig:Ring_Topology}
\end{figure}
In the case of an arbitrary interaction range $d\>1$,  the Hamiltonian $H$ in (\ref{H0}) is completely specified by matrices $R_{\ell} = R_{-\ell}^{\rT} \in \mR^{2n\x 2n}$, with $\ell = 0, \pm 1, \ldots, \pm d$ (so that $R_0 \in \mS_{2n}$) as
\begin{equation}
\label{HR}
    H
    =
    \frac{1}{2} \sum_{j=0}^{N-1}
    \Big(
    	X_j^{\rT}\sum_{\ell = -d}^d R_{\ell} X_{\mod(j-\ell,  N)}
    \Big).
\end{equation}
Here, $j-\ell$ is computed modulo $N$  in accordance with the PBCs,\footnote{with the MATLAB function being used instead of the standard modular arithmetic notation in order to avoid confusion in the subscripts} and hence, the corresponding matrix $R$ is block circulant. The matrix $R_0 \in \mS_{2n}$ specifies the free Hamiltonian for each node, while $R_{\pm s}$  describe the energy coupling between the nodes which are at a distance $s=1, \ldots, d$ from each other (with the more distant nodes in the network not being directly coupled).
In view of the block diagonal structure of the matrices $\bit{\Theta}$ and $\bit{M}$ in (\ref{xCCR}) and (\ref{MMDD}), it follows from  (\ref{BM}), (\ref{A}) and (\ref{HR}) that the QSDE (\ref{dx1}) is representable as a set of coupled QSDEs for the dynamic variables of the nodes of the network:
\begin{align}
	\label{dXj}
  		\!\!\rd X_j &= \sum_{\ell =-d}^d A_{\ell} X_{\mod(j-\ell,N)}\rd t + B \rd W_j,
  		\quad
   		 j = 0,\ldots, N-1,\!\!\!
\end{align}
where use is made of (\ref{BC}), and the matrices $A_{-d}, \ldots, A_d\in \mR^{n\x n}$ are given by
	\begin{align}
		\label{AA}
		    A_{\ell}
		    & :=
		    \left\{
		        \begin{array}{ll}
		            2 \Theta R_0 - \frac{1}{2} B J B^\rT \Theta^{-1} & {\rm if}\ \ell = 0\\
		            2 \Theta R_{\ell} & {\rm if}\ 0<|\ell| \< d
		        \end{array}
		    \right..
\end{align}
Furthermore, the output fields $Y_0, \ldots, Y_{N-1}$ in (\ref{QSDEout}), associated with the nodes of the network, evolve according to the QSDEs
\begin{align}
	\label{dYj}
  		\rd Y_j &= C X_{j}\rd t + D \rd W_j,
    \qquad
    j = 0,\ldots, N-1.
\end{align}
where $DJD^\rT$ is a submatrix of $J$ in view of the block diagonal structure of $\bit{J}$ and $\bit{D}$.
Therefore, the network, governed by (\ref{dXj}) and (\ref{dYj}), has a block circulant structure.
\section{INTERCONNECTED PLANT AND CONTROLLER NETWORKS}
\label{sec:Coherent_Quantum_Control_Arch}
Consider an interconnection of two translation invariant networks, similar to the one described in Section~\ref{sec:Ring_Topology}, where one of the networks is a quantum plant while the other is interpreted as a decentralized coherent quantum controller. We consider the case when the plant and controller networks are of equal size $N$. For example, Fig.~\ref{fig:Filter_Interconnection} provides a layout of the plant-controller network with nearest neighbour direct coupling.
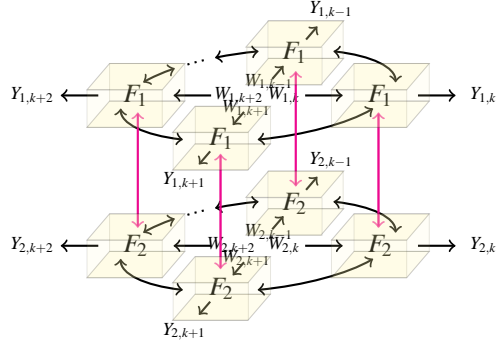
\begin{figure}[htp]
	\begin{center}
\newcommand{\myGlobalTransformationI}[2]
{
    \pgftransformcm{1}{0}{0.3}{0.35}{\pgfpoint{#1cm}{#2cm}}
}
\newcommand{\myGlobalTransformationII}[2]
{
    \pgftransformcm{1}{0}{0}{.5}{\pgfpoint{#1cm}{#2cm}}
}
\newcommand{\myGlobalTransformationIII}[2]
{
    \pgftransformcm{1}{0}{0}{.5}{\pgfpoint{#1cm}{#2cm}}
}
\newcommand{\RectanH}[2]
{
    \begin{scope}
             \myGlobalTransformationI{#1}{#2};     	
		\draw  [fill=yellow!40, opacity=.2](2.25,1) rectangle (4.25,-1);
		\draw  [fill=yellow!40, opacity=.2](-1,4.25) rectangle (1,2.25);
		\draw  [fill=yellow!40, opacity=.2](-2.25,-1) rectangle (-4.25,1);
		\draw  [fill=yellow!40, opacity=.2](1,-2.25) rectangle (-1,-4.25);		
    \end{scope}
}	
\newcommand{\RectanV}[2]
{
    \begin{scope}
             \myGlobalTransformationII{#1}{#2};     	
		\draw  [fill=yellow!40, opacity=.1](2.25,1) rectangle (4.25,-1);		
		\draw  [fill=yellow!40, opacity=.1](-2.25,-1) rectangle (-4.25,1);		
    \end{scope}
}	
\newcommand{\RectanVI}[2]
{
    \begin{scope}
             \myGlobalTransformationIII{#1}{#2};     	
		\draw  [fill=yellow!40, opacity=.1](-1,4.25) rectangle (1,2.25);
    \end{scope}
}	
\newcommand{\Connections}[3]
{
    \begin{scope}
             \myGlobalTransformationI{#1}{#2};
		\node at (3.25,0) {$F_{#3}$};
		\node at (0,3.25) {$F_{#3}$};
		\node at (-3.25,0) {$F_{#3}$};
		\node at (0,-3.25) {$F_{#3}$};
		\draw [<->, thick] ([shift=(16.5:3.5)] 0,0) arc (16.5:73.5:3.5); 
		\draw [<->, thick] ([shift=(106.5:3.5)] 0,0) arc (106.5:125:3.5); 
		\draw [dotted, thick] ([shift=(130:3.5)] 0,0) arc (130.5:140:3.5);
		\draw [<->, thick] ([shift=(144:3.5)] 0,0) arc (144:163.5:3.5);
		\draw [<->, thick]([shift=(196.5:3.5)] 0,0) arc (196.5:253.5:3.5); 
		\draw [<->, thick]([shift=(-16.5:3.5)] 0,0) arc (-16.5:-73.5:3.5); 
    \end{scope}
}	
\newcommand{\myBoxes}[2]
{
    \begin{scope}
             	\RectanH{#1}{#2cm-.5};
		\RectanH{#1}{#2cm+.5};		
		\RectanV{.3}{#2cm+0.35};
		\RectanV{-.3}{#2cm-0.35};
		\RectanVI{1.26}{#2cm-0.14};
		\RectanVI{-.67}{#2cm-2.42};
		\RectanVI{-1.26}{#2cm-3.1};
		\RectanVI{.7}{#2cm-.85};		
    \end{scope}
}	
\newcommand{\Interonnections}[2]
{
    \begin{scope}
        \myGlobalTransformationI{#1}{#2};     	 	             		
				\node at (0,1) {\scriptsize $W_{1,k-1}$};
		\node at (0.7,0) {\scriptsize $W_{1,k}$};
		\node at (0,-1) {\scriptsize $W_{1,k+1}$};
		\node at (-.5,0) {\scriptsize $W_{1,k+2}$};

		\draw [->, thick] (1.25,0) -- (2.25,0);
		\draw [->, thick] (-1.25,0) -- (-2.25,0);
		\draw [->, thick] (0,1.25) -- (0,2.25);
		\draw [->, thick] (0,-1.25) -- (0,-2.25);		
		
		\draw [->, thick] (4.25,0) -- (5.25,0);
		\draw [->, thick] (-4.25,0) -- (-5.25,0);
		\draw [->, thick] (0,4.25) -- (0,5.25);
		\draw [->, thick] (0,-4.25) -- (0,-5.25);
		
		\node at (0,6.5) {\scriptsize $Y_{1, k-1}$};
		\node at (-6.,0) {\scriptsize $Y_{1, k+2}$};
		\node at (0,-6.5) {\scriptsize $Y_{1, k+1}$};
		\node at (6.,0) {\scriptsize $Y_{1, k}$};					
    \end{scope}
}	
\newcommand{\InteronnectionsP}[2]
{
    \begin{scope}
        \myGlobalTransformationI{#1}{#2};     	 	             		
		\node at (0,1) {\scriptsize $W_{2, k-1}$};
		\node at (0.7,0) {\scriptsize $W_{2, k}$};
		\node at (0,-1) {\scriptsize $W_{2, k+1}$};
		\node at (-.7,0) {\scriptsize $W_{2, k+2}$};

		\draw [->, thick] (1.25,0) -- (2.25,0);
		\draw [->, thick] (-1.25,0) -- (-2.25,0);
		\draw [->, thick] (0,1.25) -- (0,2.25);
		\draw [->, thick] (0,-1.25) -- (0,-2.25);
		
		\draw [->, thick] (4.25,0) -- (5.25,0);
		\draw [->, thick] (-4.25,0) -- (-5.25,0);
		\draw [->, thick] (0,4.25) -- (0,5.25);
		\draw [->, thick] (0,-4.25) -- (0,-5.25);	
		
		\node at (0,6.5) {\scriptsize $Y_{2, k-1}$};
		\node at (-6.,0) {\scriptsize $Y_{2, k+2}$};
		\node at (0,-6.5) {\scriptsize $Y_{2, k+1}$};
		\node at (6.,0) {\scriptsize $Y_{2, k}$};		
    \end{scope}
}	
\newcommand{\InteronnectionsPC}[4]
{
    \begin{scope}
         \myGlobalTransformationIII{#1}{#2};     	 	             						
		\draw [opacity=.7,<->, thick,color=magenta] (0,#3cm) -- (0,-5.25);					
    \end{scope}
}	
\begin{tikzpicture}[scale=.5]
\Connections{0}{0}{2}
\Connections{0}{4}{1}
\InteronnectionsP{0}{0}
\Interonnections{0}{4}
\InteronnectionsPC{-1}{2}{.9}{+1}
\InteronnectionsPC{1}{4.2}{.9}{-1}
\InteronnectionsPC{3.2}{3.1}{0.9}{}
\InteronnectionsPC{-3.2}{3.1}{0.9}{+2}
\myBoxes{0}{0}
\myBoxes{0}{4}
\end{tikzpicture}
	\end{center}
	\vskip-3mm 
\caption{Finite fragments of two translation invariant quantum networks of equal size with nearest neighbour direct coupling. The 1st network is a plant and the 2nd is a controller. Also shown are the input and output fields of the nodes of the networks.}
\label{fig:Filter_Interconnection}
\end{figure}
In what follows, the operators and parameters of the plant and controller networks will be indicated by subscripts 1 and 2, respectively. In particular,
$X_{1,j}$ and $X_{2,j}$ denote the vectors of dynamic variables for the $j$th nodes of the plant and the controller, respectively, with the corresponding dimensions $2n_1$ and $2n_2$. The total Hamiltonian of the plant-controller interconnection is
\begin{equation}
\label{H_total}
	H := H_1 + H_2 + H_{12}.
\end{equation}
Here, in view of (\ref{HR}), the plant Hamiltonian $H_1$ and the controller Hamiltonian $H_2$ are given by
\begin{equation}
	\label{HTT}
	\!H_k := \frac{1}{2} \sum_{j=0}^{N-1}
		   \Big(
					X_{k,j}^{\rT}
					\sum_{\ell = -d_k}^{d_k} R_{k,\ell} X_{k, \mod(j-\ell, N)}
		   \Big),
\qquad
k=1,2,\!\!\!\!
\end{equation}
where
$R_{k, \ell} = R_{k, -\ell}^{\rT} \in \mR^{2n_k\x 2n_k}$, with $\ell = 0, \pm 1, \ldots, \pm d_k$. Also, $H_{12}$ is the interaction Hamiltonian:
\begin{equation}
	\label{H12}
	H_{12}:= \sum_{j=0}^{N-1}
			  \Big(
			 	 X_{1,j}^{\rT}
			  	\sum_{\ell = -\wt{d}}^{\wt{d}}
			  		\wt{R}_{\ell} X_{2, \mod(j -\ell, N)}
			  \Big),
\end{equation}
where $\wt{d}$ denotes the range of direct coupling between the plant and controller nodes (with $N> 2\max(d_1, d_2, \wt{d})$), and $\wt{R}_{\ell} \in \mR^{2n_1 \times 2n_2}$ are the plant-controller coupling matrices, with $|\ell| \< \wt{d}$. It is assumed that the plant variables commute with the controller variables: $[X_{1,j}, X_{2,k}^{\rT}] = 0$ for all $0\< j,k < N$. Hence, $H_{12}$ in (\ref{H12}) is indeed a self-adjoint operator.  Due to the structure of the total Hamiltonian, the plant-controller network is a block circulant system \cite{Wall1978}. By substituting $H$ specified by (\ref{H_total})--(\ref{H12}) into the QSDE (\ref{dx}), it follows that the plant and controller variables are governed by a set of coupled QSDEs:
\begin{align}
		\nonumber				
		\rd X_{1,j}
		=&
    \Big(
		 \sum_{\ell =-d_1}^{d_1}
		 	A_{1,\ell} X_{1,\mod(j-\ell,N)}
		+
		\sum_{\ell = -\wt{d}}^{\wt{d}}
			\wt{A}_{1,\ell} X_{2,\mod(j-\ell,N)}
\Big)\rd t
		\\
\label{dX1j}
		&+ B_1 \rd W_{1,j},\\
\nonumber
		\rd X_{2,j}	
		=&
    \Big(
        \sum_{\ell = -\wt{d}}^{\wt{d}}
				\wt{A}_{2,\ell} X_{1,\mod(j+\ell,N)}
		+
			\sum_{\ell =-d_2}^{d_2}
				A_{2,\ell} X_{2,\mod(j-\ell,N)}
\Big)\rd t
		\\
\label{dX2j}
		& + B_2 \rd W_{2,j}
\end{align}\noindent
for all $0\< j< N$,
where, similarly to (\ref{BC}) and (\ref{AA}), the matrices of coefficients are given by
\begin{align}
	\label{AAk}
    A_{k,\ell}
    & :=
    \left\{
        \begin{array}{ll}
            2 \Theta_k R_{k,0} - \frac{1}{2} B_k J_k B_k^\rT \Theta_k^{-1} & {\rm if}\ \ell = 0
            \\
            2 \Theta_k R_{k,\ell} & {\rm if}\ 0< |\ell|\< d_k
        \end{array}
    \right.,\\
    \label{tAA2}
	\wt{A}_{1,\ell}&:= 2\Theta_1 \wt{R}_{\ell},
    \quad
	\wt{A}_{2,\ell}:= 2\Theta_2 \wt{R}_{\ell}^\rT,
    \quad |\ell|\< \wt{d}, \\
\label{BBk}
    B_k & := 2 \Theta_k M_k^{\rT }.
\end{align}
The coupled QSDEs (\ref{dX1j})--(\ref{dX2j}), which are  associated with nodes of the plant and the controller, can be studied in the spatial frequency domain. Similarly to \cite{VP2014}, we will use the discrete Fourier transforms (DFTs) of the quantum processes $X_{k,j}$ and $W_{k,j}$ over the spatial subscripts $j=0, ..., N-1$ for $k=1,2$:
\begin{equation}
	\label{ZtranX_ZtranW}
	\!\cX_{k,z}(t)
    :=\sum_{j=0}^{N-1} z^{-j}X_{k,j}(t),
    \qquad
	\cW_{k,z}(t):=\sum_{j=0}^{N-1} z^{-j}W_{k,j}(t),	\!\!\!\!
\end{equation}
which are considered for $z\in \mU_N$ from the set of $N$th roots of unity
\begin{equation}
	\label{mUN}
	\mU_N
    :=
    \big\{
        \re^{\frac{ 2\pi i j}{N}}: \
        j=0,\ldots,N-1
    \big\}.
\end{equation}
Then the corresponding augmented quantum process
$
\cX_z
$
is evolved in time by the QSDE
\begin{align}
	\label{ZtranSys}	
	\rd \cX_z  &=  \cA_z \cX_z \rd t+ \cB \rd \cW_z,
    \quad
    \cX_z :=
	{\small\begin{bmatrix}
		\cX_{1,z} \\  \cX_{2,z}
	\end{bmatrix}},
\ \
    \cW_z:={\small\begin{bmatrix} \cW_{1,z} \\ \cW_{2,z} \end{bmatrix}},
\end{align}
where the matrices  $A_z$ and $\cB$ are defined in terms of the network parameters in (\ref{AAk})--(\ref{BBk}) as
\begin{equation}
		\label{cAz}
		\cA_z:=
		{\small\begin{bmatrix}
				\scriptsize
				\sum_{\ell=-d_1}^{d_1} z^{-\ell} A_{1,\ell}  		
				& 			
				\sum_{\ell=-\wt{d}}		
				^{\wt{d}} z^{-\ell}\wt{A}_{1,\ell}
				\\
				\sum_{\ell=-\wt{d}}^{\wt{d}} z^{\ell}\wt{A}_{2,\ell} 		
				&		
				\sum_{\ell=-d_2}^{d_2} z^{-\ell}A_{2,\ell}
		\end{bmatrix}},
    \quad
	\cB
    :=
	{\small\begin{bmatrix}
		B_1 & 0 \\ 0 & B_2
	\end{bmatrix}},
\end{equation}
see \cite{SVP_TI_2015} for more details.
Here, use is made of (\ref{ZtranX_ZtranW}) together with the PBCs and the well-known properties of DFTs due to the assumption that $z\in \mU_N$. As discussed in \cite{VP2014}, it follows from the CCRs (\ref{xCCR}) that
\begin{align}
\nonumber	
    \!\!\!\!\![\cX_z, \cX_v^\dagger]
    & =
    \sum_{j,k=0}^{N-1}
    z^{-j} v^k [X_j,X_k^\rT]\\
\label{XX}
						& =2i\Theta \sum_{k=0}^{N-1}\Big(\frac{v}{z}\Big)^k
						=2i\delta_{zv}N\Theta,
    \qquad
    X_j:= {\small\begin{bmatrix}X_{1,j} \\ X_{2,j}\end{bmatrix}}\!\!\!
\end{align}
for all $z,v\in \mU_N$, where $\Theta:={\scriptsize \begin{bmatrix}\Theta_1 & 0\\ 0 & \Theta_2\end{bmatrix}}$ is the common CCR matrix of the augmented vectors  $X_j$ for all $j=0, \ldots, N-1$. By a similar reasoning, (\ref{WW}) implies that
\begin{align}
	\label{Ito_net}
	\rd \cW_z \rd \cW_v^\dagger = N \delta_{zv} \Omega \rd t,
    \qquad
    \Omega :={\small\begin{bmatrix}\Omega_1 & 0\\ 0 & \Omega_2\end{bmatrix}},
\end{align}
where $\Omega_1$ and $\Omega_2$ are the quantum It\^{o} matrices of the plant and controller Wiener processes $W_{1,j}$ and $W_{2,j}$. Assuming the network size $N$ to be fixed, we denote the matrix of second-order cross-moments of the quantum processes $\cX_z$ and $\cX_v$ in (\ref{ZtranX_ZtranW}) by
	\begin{equation}
		\label{DFTcov}
		\cS_{z,v}(t)
        :=
        \bE
        (
            \cX_z(t)
            \cX_v(t)^{\dagger}
        ).	
	\end{equation}
Here and in what follows,  the quantum expectation $\bE(\cdot)$ is over the tensor product $\rho:= \varpi \ox \upsilon$  of the initial state $\varpi$ of the plant-controller network and the vacuum state $\upsilon$ of the external fields.

\begin{lem}
\label{lem:cov_net_time}
Suppose the matrix $\cA_z$ in (\ref{cAz}) is Hurwitz for any $z\in \mU_N$ from (\ref{mUN}). Then for all $z,v \in \mU_N$, the matrix in (\ref{DFTcov}) has the limit
$    \cS_{z,v}(\infty)
    :=
    \lim_{t\to +\infty}
    \cS_{z,v}(t)
    =
    N \delta_{zv} \cS_z
$, 
where the matrix
$\cS_z=\cS_z^*\succcurlyeq 0$ is a unique solution of the ALE
\begin{equation}
\label{cSz}
    \cA_z \cS_{z} + \cS_{z} \cA_z^* + \cB \Omega \cB^\rT=0.
\end{equation}
\end{lem}

The proof of Lemma~\ref{lem:cov_net_time} can be found in \cite{SVP_TI_2015} and is based on linearity of the QSDE (\ref{ZtranSys}), the relation (\ref{Ito_net}) and the fact that the forward increments of the quantum Wiener process in the vacuum state are uncorrelated with the adapted processes.
The matrix-valued function $\mU \ni z\mapsto \cS_z$ in (\ref{cSz}) (which is well-defined on the unit circle under a stronger stability condition formulated in the next section) is the spatial spectral density \cite{VP2014} of the closed-loop system. This density encodes the covariance structure of the plant-controller variables in the invariant Gaussian quantum state \cite{KRP_2010} in the limit of infinite network size $N\to +\infty$. In this sense, $\cS_z$ is a network counterpart to the controllability Gramian \cite{AM_1989} of the system.
\section{OPTIMAL COHERENT QUANTUM CONTROL PROBLEM}\label{sec:Coherent_Quantum_Control}

In view of the block circulant structure of the plant-controller network of Section~\ref{sec:Coherent_Quantum_Control_Arch}, this closed-loop system is stable for any network size $N\> 1$ if and only if the matrix $\cA_z $ in (\ref{cAz}) is Hurwitz for all $z \in \mU$, that is,
\begin{equation}
	\label{stab}
	\max_{z \in \mU} \br(\text{e}^{\cA_z}) < 1.
\end{equation}
Here, use is also made of continuity of $\cA_z$ with respect to $z\in \mU$ and the property that  $\bigcup_{N=1}^{+\infty}\mU_N$ is a dense subset of $\mU$. In what follows, the plant-controller network is called \emph{stable} if it satisfies (\ref{stab}). In this case, the decentralized coherent quantum controller is referred to as a \emph{stabilizing} controller. As a performance criterion for such controllers, we will use a steady-state weighted mean square cost functional which has to be minimized in the limit of infinite network size. 
For simplicity, it is assumed that each node of the controller is only coupled to one node of the plant, that is, $\wt{d}=0$.  
The performance cost functional is given by
\begin{equation}
	\label{cost_fun}
	\cE_N
    :=
    \frac{1}{N}
    \lim_{t \to +\infty}
    \bE
    \sum_{j,k=0}^{N-1}
    	\wt{X}_j(t)^\rT 	
		\sigma_{j-k} 				
		\wt{X}_k(t),
\end{equation}
where $\wt{X}_j$ are auxiliary quantum processes associated with $X_j$ from (\ref{XX}) by
\begin{equation}
\label{XE}
	\wt{X}_j(t)
	:=
	E
    X_j,
	\qquad
	E:=
	{\small\begin{bmatrix}
			I_{2n_1} & 0
			\\
			0 & \wt{R}_0
	\end{bmatrix}},
\end{equation}
for $j=0, \ldots, N-1$, and $\wt{R}_0$ is the coupling matrix which, in accordance with (\ref{H12}), specifies the interaction Hamiltonian $	H_{12}= \sum_{j=0}^{N-1}
			 	 X_{1,j}^{\rT}
			  		\wt{R}_0 X_{2, j}
$  in the case $\wt{d}=0$.
Moreover, $\sigma_k$ is a given $\mR^{4n_1 \x 4n_1}$-valued sequence which satisfies $\sigma_{−-k} = \sigma_k^\rT$
for all integers $k$ and specifies a real symmetric block Toeplitz weighting matrix $(\sigma_{j-k})_{0\<j,k<N}$. The block Toeplitz structure of the weighting matrix in (\ref{cost_fun}) corresponds to the translation invariance of the quantum network being considered.
A matrix-valued map $\mU \ni z \mapsto \Sigma_z = \Sigma_z^*$, defined as the Fourier transform
\begin{equation}
	\label{Sigmaz}
	\Sigma_z:=\sum_{k=-\infty}^{+\infty} z^{-k} \sigma_k,
\end{equation}
describes the spectral density of the weighting sequence. In order to ensure the absolute convergence of the series in (\ref{Sigmaz}), it is assumed that $\sum^{+\infty}_{k=-\infty} \|\sigma_k\| < +\infty$, which also makes $\Sigma_z$ a continuous function of $z$. The fulfillment of the condition $\Sigma_z \succcurlyeq 0$ for all $z \in \mU$ is necessary and sufficient for $(\sigma_{j-k})_{0\<j,k<N} \succcurlyeq 0$ to hold  for all $N\>1$; see, for example \cite{gren58}. In this case, the sum on the right-hand side of (\ref{cost_fun}) is a positive semi-definite operator, and hence, $\cE_N \>0$.
The cost functional $\cE_N$ in (\ref{cost_fun}) resembles the classical LQG control performance index \cite{AM_1989} which imposes a cost on the actuation signal. A similar approach can be used in order to ``penalize'' other variables of interest by an appropriate choice of the weighting matrices $\sigma_k$ in (\ref{cost_fun}) and the matrix  $E$ in (\ref{XE}).
For any stabilizing controller and any given (sufficiently large) network size $N$, Lemma~\ref{lem:cov_net_time} allows the cost functional in (\ref{cost_fun}) can be computed as \cite{SVP_TI_2015}:
\begin{equation}	
\label{cEN}
	\cE_N
		 =
    \frac{1}{N}
    \sum_{z \in \mU_N}
    \big\bra
        \wh{\Sigma}_N(z),\,
        E \cS_z E^{\rT}
    \big\ket,
\end{equation}
where the matrix $\cS_z$ is the unique solution of the ALE (\ref{cSz}), and
$	\wh{\Sigma}_N(z)
    :=
    \sum_{|k|<N}
    \big(
        1-\frac{|k|}{N}
    \big)
    z^{-k}\sigma_k
$. 
Due to the uniform convergence $\lim_{N\to +\infty} \wh{\Sigma}_N(z) = \Sigma_z$ to the spectral density (\ref{Sigmaz}) over $z \in \mU$, the representation (\ref{cEN}) leads to  the following infinite network size limit of the mean square cost functional $\cE_N$ in (\ref{cost_fun}):
\begin{align}
\nonumber
	\cE
      :=	&
	\lim_{N \to+\infty}
    \cE_N
    =
    \frac{1}{2 \pi i}
    \oint_\mU
    \bra
        \Sigma_z,\,
        E \cS_z E^\rT
    \ket
    \frac{\rd z}{z}\\
\label{cE}
    =&  
    \frac{1}{2 \pi}
    \int_0^{2\pi}
    \big\bra
        \Sigma_{\re^{i\varphi}},\,
        E \cS_{\re^{i\varphi}} E^\rT
    \big\ket
    \rd \varphi
    =
    \Res_{z=0}
    \frac{\bra
        \Sigma_z,\,
        E \cS_z E^\rT
    \ket}{z}
\end{align}
for any stabilizing controller in the sense of (\ref{stab}); see \cite{VP2014,SVP_TI_2015} for more details. The resulting cost functional $\cE$ in (\ref{cE}) corresponds to the thermodynamic limit of equilibrium statistical mechanics \cite{Ruelle78}.
We will now consider a decentralized coherent quantum control problem which is formulated as the minimization
\begin{equation}
\label{cEmin}
  \cE\longrightarrow \min
\end{equation}
of the infinite network cost functional in (\ref{cE}) over the energy and coupling matrices of stabilizing controllers (with all the dimensions of individual nodes being fixed). A particular case of such a network with nearest neighbour interaction (when $d_1 = d_2=1$ and $\wt{d}=0$) is depicted in Fig.~\ref{fig:chain_direct_coupling}.
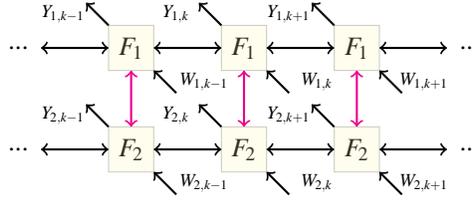
\begin{figure}[htp]
\begin{center}
\newcommand{\Chain}[2]
{
    \begin{scope}
		\node at (0,#1) {$F_{#2}$};
		\node at (-5,#1) {$F_{#2}$};
		\node at (5,#1) {$F_{#2}$};
		\node at (10,#1) {...};
		\node at (-10,#1) {...};
		\draw [fill=yellow!40, opacity=.2] (-1,-1+#1) rectangle (1,1+#1);
		\draw  [fill=yellow!40, opacity=.2] (-6,-1+#1) rectangle (-4,1+#1);
		\draw  [fill=yellow!40, opacity=.2] (4,-1+#1) rectangle (6,1+#1);
		\draw [<->, thick] (-1,#1) -- (-4,#1);
		\draw [<->, thick] (1,#1) -- (4,#1);
		\draw [<->, thick] (6,#1) -- (9,#1);
		\draw [<->, thick] (-9,#1) -- (-6,#1);
		\draw [->, thick] (2,-2+#1) -- (1,-1+#1);
		\draw [->, thick] (7,-2+#1) -- (6,-1+#1);
		\draw [->, thick] (-3,-2+#1) -- (-4,-1+#1);
		
		\draw [->, thick] (-1,1+#1) -- (-2,2+#1);
		\draw [->, thick] (-6,1+#1) -- (-7,2+#1);
		\draw [->, thick] (4,1+#1) -- (3,2+#1);
		
		\node at (3.25,-1.5+#1) {\scriptsize $W_{#2, k}$};
		\node at (8,-1.5+#1) {\scriptsize $W_{#2, k+1}$};
		\node at (-1.75,-1.5+#1) {\scriptsize $W_{#2, k-1}$};	
		
		\node at (-3,1.5+#1) {\scriptsize $Y_{#2, k}$};
		\node at (-8,1.5+#1) {\scriptsize $Y_{#2, k-1}$};
		\node at (2,1.5+#1) {\scriptsize $Y_{#2, k+1}$};		
    \end{scope}
}	
\begin{tikzpicture}[scale=.3]
		\Chain{0}{2}
		\Chain{4.5}{1}
		
		\draw [color=magenta,<->, thick] (5,1+2.5) -- (5,1);
		\draw [color=magenta,<->, thick] (0,1+2.5) -- (0,1);
		\draw [color=magenta,<->, thick] (-5,1+2.5) -- (-5,1);		
\end{tikzpicture}
\end{center}
\vskip-3mm 
\caption{The infinite plant-controller network with nearest neighbour direct coupling. Also shown are the input and output fields.}
\label{fig:chain_direct_coupling}
\end{figure}
With the matrix $ B_2 \in \mR^{2n_2 \x 2m_2} $ in (\ref{BBk}) being fixed, the cost $\cE$ in (\ref{cE}) is a function of
\begin{align}
\nonumber
     g
     & := 
     (R_{2,0},R_{2,1},\ldots,R_{2,d_2},\wt{R}_0) \\ 
\label{g}
     & \in 
     \mS_{2n_2}
     \x
  	   (\mR^{2n_2 \times 2n_2})^{d_2}
     \x
     \mR^{2n_1 \x 2n_2}
     =:
     \mG
\end{align}
which parameterizes the controller matrices in (\ref{AAk})--(\ref{tAA2}). The minimization of the cost functional $\cE$ in (\ref{cEmin}) is carried out over the set
\begin{equation}
    \label{G0}
        \mG_0:= \{g \in \mG:\ \cA_z \ {\rm in}\ (\ref{cAz})\ {\rm satisfies}\ (\ref{stab}) \}
\end{equation}
of those $g$ which specify stabilizing quantum controllers (\ref{dX2j}) for the quantum plant (\ref{dX1j}). The set $\mG$ on the right-hand side of (\ref{g}) is endowed with the structure of a Hilbert space with the direct sum inner product
$	\bra
		g
		,
		g'
	\ket
	:=
\bra
		\wt{R}_0
		,
		\wt{R}'_0
		\ket
+
	\sum_{\ell=0}^{d_2}
		\bra
			R_{2,\ell},R_{2,\ell} '
		\ket
$.
In what follows, we will use an auxiliary spectral density $\cQ_z$ defined on the unit circle $z\in \mU$  as the unique solution of the ALE 
\begin{align}
		\label{cQz}
		\cA_z^* \cQ_z + \cQ_z \cA_z + E^\rT \Sigma_z E=0
\end{align}
for a stabilizing controller, where use is made of (\ref{XE}) and (\ref{Sigmaz}). The function $\cQ_z$ corresponds to the observability Gramian \cite{AM_1989}. We will also use a network counterpart to the Hankelian \cite{VP_2013a}:
\begin{align}
		\label{cHz}
		\cH_z := \cQ_z \cS_z,
\end{align}
where $\cS_z$ is given by (\ref{cSz}). These and related matrices are partitioned into four blocks $(\cdot)_{jk}$ indexed by $1\< j,k\< 2$ according to their association with the plant and controller variables, and the Fourier parameter $z$ will sometimes be omitted.  
The following theorem, which  is formulated for the case $\wt{d}=0$,  can be extended to arbitrary range of the plant-controller coupling.

\begin{thm}
\label{thm:NecCon}
Suppose the plant-controller coupling range in the network is $\wt{d}=0$. Then necessary conditions of optimality for a stabilizing controller in the problem (\ref{cEmin}) are as follows:
		\begin{align}
               	\label{Nec_Rl}
               	&\Res_{z=0}
               \frac{
					\Re                	
                	(
                        z^{\ell}
                        (
                		\Theta_2 \cH_{22}
                        -
                        \cH_{22}^* \Theta_2
                        )
               		)}
               		{z}
                =0,
               	\qquad
               \ell=0, \ldots, d_2,\\
\label{Nec_Rtl}
                &\Res_{z=0}
                \frac{
					\Re				
					(
						\cH_{21}^* \Theta_2
                         -\Theta_1 \cH_{12}
						+
                    \frac{1}{2}
						(						
						\Sigma_z
						E
						\cS_z
						)_{22}
					)}{z}
         = 0.
       \end{align}
\end{thm}
\begin{proof}
The matrices $R_{2,\ell}=R_{2,-\ell}^{\rT}$ in (\ref{AAk}) influence the cost functional $\cE$ in (\ref{cE}) only through the matrix $\cS_z$ in (\ref{cSz}) which depends on those matrices through $\cA_z$ in (\ref{cAz}). Hence, the first variation of the integrand in (\ref{cE}) with respect to $R_{2,\ell}$ is
\vskip-8mm
\begin{align}
\nonumber
		\delta_{R_{2,\ell}}&  
    \bra
        \Sigma_z,\,
        E \cS_z E^{\rT}
    \ket
        =
    \bra
        E^{\rT}\Sigma_z E,\,
        		\delta_{R_{2,\ell}}\cS_z
    \ket\\
\nonumber
        &=
		-\bra \cA_z^* \cQ_z + \cQ_z \cA_z  ,\, \delta_{R_{2,\ell}} \cS_z \ket	\\
\nonumber
		&=
		-\bra \cQ_z  , \cA_z \delta_{R_{2,\ell}} \cS_z + (\delta_{R_{2,\ell}} \cS_z) \cA_z^*\ket	\\
\label{dR1}
		&=
		\bra \cQ_z  , (\delta_{R_{2,\ell}} \cA_z)  \cS_z + \cS_z \delta_{R_{2,\ell}} \cA_z^*\ket	
		=
		2\Re
		\bra \cH_z  , \delta_{R_{2,\ell}} \cA_z \ket
\end{align}\noindent	
for any $z\in \mU$ and $\ell = 0, \ldots, d_2$, provided the controller parameters in (\ref{g}) satisfy $g \in \mG_0$ in (\ref{G0}). 
Here, use is made of the ALEs (\ref{cSz}), (\ref{cQz}) and the Hankelian $\cH_z$ from (\ref{cHz}), and
\begin{equation}
\label{dR2}
    \delta_{R_{2,\ell}}
    \cA_z
    \!=\!
    {\small\begin{bmatrix}0 & 0\\ 0 & 2\end{bmatrix}}
    \ox
    \left\{\!\!\!\!
    {\small\begin{array}{ll}
        \Theta_2\delta R_{2,0} & {\rm if}\ \ell = 0\\
        \Theta_2(z^{-\ell} \delta R_{2,\ell} + z^{\ell} \delta R_{2,\ell}^{\rT})
        & {\rm if}\ \ell> 0
    \end{array}}
    \right.\!\!\!.\!\!\!
\end{equation}\noindent	
Substitution of (\ref{dR2}) into (\ref{dR1}) yields the following Fr\'{e}chet derivatives
\begin{equation}
\label{dR3}
		\!\!\d_{R_{2,\ell}}
    \bra
        \Sigma_z,
        E \cS_z E^{\rT}
    \ket
		\!=\!
		\left\{\!\!\!\!\!\!
        {\small\begin{array}{ll}				
		2\Re
        (
            \cH_{22}^* \Theta_2 - \Theta_2 \cH_{22}
        )
		&
		{\rm if}\ \ell=0		
		\\
        4
		\Re
        (
		z^{\ell}
        (\cH_{22}^* \Theta_2- \Theta_2 \cH_{22})
        )
		&
        {\rm if}\
		\ell > 0
		\end{array}}
		\right.\!\!\!\!, \! \! \!\!\!
	\end{equation}\noindent
where the symmetry of $R_{2,0}$ is taken into account. Upon integration according to (\ref{cE}), it follows from (\ref{dR3}) that 
\begin{equation}
\label{dEdR1}
		\d_{R_{2,\ell}}
    \cE 
		=
		\left\{
        {\small\begin{array}{ll}				
		2\Res_{z=0}
    \frac{\Re
        (
            \cH_{22}^* \Theta_2 - \Theta_2 \cH_{22}
        )}{z}
		&
		{\rm if}\ \ell=0		
		\\
        4
		\Res_{z=0}
        \frac{\Re
        (
		z^{\ell}
        (\cH_{22}^* \Theta_2- \Theta_2 \cH_{22})
        )}{z}
		&
        {\rm if}\
		\ell > 0
		\end{array}}
		\right..
	\end{equation}\noindent
Now, the plant-controller coupling matrix $\wt{R}_0$  influences the cost functional $\cE$ not only through the matrices  (\ref{tAA2}), (\ref{cAz}) and the ALE (\ref{cSz}), but also through the matrix $E$ in (\ref{XE}). Therefore, by appropriately modifying (\ref{dR1}), it follows that
\begin{equation}
\label{dR4}
		\!\delta_{\wt{R}_0}
    \bra
        \Sigma_z,\,
        E \cS_z E^{\rT}
    \ket
        \!=\!
		2\Re
        (
		\bra \cH_z  , \delta_{\wt{R}_0} \cA_z \ket
+
    \bra
        \Sigma_z E \cS_z,\,
        		\delta_{\wt{R}_0}E
    \ket),\!\!\!\!
\end{equation}\noindent	
which holds for any stabilizing  controller (with $g\in \mG_0$), where 
\begin{equation}
\label{dR5}
    \delta_{\wt{R}_0}
    \cA_z
    =
    2
    {\small\begin{bmatrix}
        0 & \Theta_1 \delta \wt{R}_0\\
        \Theta_2 \delta \wt{R}_0^{\rT} & 0
    \end{bmatrix}},
    \qquad
    \delta_{\wt{R}_0} E = {\small\begin{bmatrix}0 & 0 \\ 0 & \delta \wt{R}_0\end{bmatrix}}.    
\end{equation}
Substitution of (\ref{dR5}) into (\ref{dR4}) leads to the corresponding Fr\'{e}chet derivative of the integrand in (\ref{cE}):
$
		\d_{\wt{R}_0}
    \bra
        \Sigma_z,\,
        E \cS_z E^{\rT}
    \ket
    =
	4\Re				
	\big(
		\cH_{21}^* \Theta_2
                     -\Theta_1 \cH_{12}
		+
                \frac{1}{2}
		(						
		\Sigma_z
		E
		\cS_z
		)_{22}
	\big)
$, 
and hence, 
\begin{equation}
\label{dEdR2}
		\d_{\wt{R}_0}\cE
=
                4\Res_{z=0}
                \frac{
					\Re				
					(
						\cH_{21}^* \Theta_2
                         -\Theta_1 \cH_{12}
						+
                    \frac{1}{2}
						(						
						\Sigma_z
						E
						\cS_z
						)_{22}
					)}{z}.
\end{equation}
The necessary conditions of optimality (\ref{Nec_Rl}) and (\ref{Nec_Rtl}) can now be obtained by equating to zero the Fr\'{e}chet derivatives in (\ref{dEdR1}) and (\ref{dEdR2}), respectively.
\end{proof}
The optimality conditions (\ref{Nec_Rl}) and (\ref{Nec_Rtl}) provide a set of algebraic equations for finding an optimal controller in the problem (\ref{cEmin}). Although their solution is not yet available, the Fr\'{e}chet derivatives of the cost functional in (\ref{dEdR1}) and (\ref{dEdR2}) can be used, for example, in a gradient descent numerical algorithm for finding locally optimal stabilizing controllers, similar to \cite{SVP_GD_2015}.
\section{CONCLUSION} \label{sec:conclusion}
We have considered a decentralized coherent quantum control problem for a quantum plant  modelled as a large-scale one-dimensional translation invariant network governed by linear QSDEs with PBCs. The problem is to design a similar coherent quantum controller network, directly coupled to the plant, so as to minimize a weighted mean square functional in the thermodynamic limit.
We have used DFTs in order to achieve a more tractable form of the problem in the spatial frequency domain.   This reformulation has allowed necessary conditions of optimality to be obtained for a stabilizing controller. Using a similar approach, the results of the paper can be extended to translation invariant interconnections of linear quantum stochastic systems on higher dimensional lattices.

\end{document}